\documentclass[copyright,creativecommons]{eptcs}
\providecommand{\event}{QPL 2012}
\def\authorrunning{R. Hermens}
\def\titlerunning{Speakable in quantum mechanics: babbling on}

\usepackage{amsmath}
\usepackage{amssymb}
\usepackage{amsthm}
\usepackage{dsfont}
\usepackage{color} 
\DeclareMathOperator{\h}{\mathcal{H}} 
\DeclareMathOperator{\A}{\mathcal{A}} 
\DeclareMathOperator{\pee}{\mathds{P}}
\DeclareMathAlphabet{\mathpzc}{OT1}{pzc}{m}{it}   
\newtheorem{lemma}{Lemma}[section] 
\newtheorem{theorem}{Theorem}[section]
\theoremstyle{definition}
\newtheorem{definition}{Definition}[section] 

\begin{document}

\title{Speakable in quantum mechanics: babbling on} 
\author{Ronnie Hermens\institute{Department of Theoretical Philosophy\\University of Groningen}\email{r.hermens@rug.nl}}

\maketitle

\begin{abstract}
This paper summarizes and elaborates on some of the results in \cite{Hermens12}.
The summary consists of a short version of the derivation of the intuitionistic quantum logic $L_{QM}$ (which was originally introduced in \cite{CHLS09}). 
The elaboration consists of extending this logic to a classical logic $CL_{QM}$. 
Some first steps are then taken towards setting up a probabilistic framework based on $CL_{QM}$ in terms of R\'enyi's conditional probability spaces.
Comparisons are then made with the traditional framework for quantum probabilities.
\end{abstract}

\section{Introduction}
In the Hilbert space formalism of quantum mechanics (QM), a quantum probability function is a function $\pee$ from the lattice $L(\h)$ of projection operators acting on the Hilbert space $\h$, to the unit interval $[0,1]$, which satisfies the rules
\begin{itemize}
\item $\pee(1)=1$,
\item $\pee(P_1\vee P_2\vee\ldots)=\pee(P_1)+\pee(P_2)+\ldots$ \quad whenever $P_iP_j=0$ for all $i\neq j$. 
\end{itemize}
This characterization of quantum probabilities is reminiscent (at least in form) of the classical structure, where a probability function is specified by a pair $(\mathcal{F},\pee)$ where $\mathcal{F}$ denotes a $\sigma$-algebra of subsets of some space $\Omega$ and $\pee:\mathcal{F}\to[0,1]$ is a function which satisfies
\begin{itemize}
\item $\pee(\Omega)=1$,
\item $\pee(\Delta_1\cup\Delta_2\cup\ldots)=\pee(\Delta_1)+\pee(\Delta_2)+\ldots$ \quad whenever $\Delta_i\cap\Delta_j=\varnothing$ for all $i\neq j$.
\end{itemize}

Because of the resemblance in structure between the two definitions, it is tempting to think that one must be able to find resemblances between admissible interpretations of the probabilities.
However, as indicated in \cite{Wilce12}, such a program has to face the problem of explaining the non-Boolean nature of $L(\h)$.
This issue may best be understood by contrasting the quantum structure with the classical structure of $\sigma$-algebras.
By far the majority of philosophical approaches to classical probability agree on one thing.
Namely, that the probability $\pee(\Delta)$ can be understood as the probability that some proposition $S_\Delta$ codified by $\Delta$ is true. 
For example, if $\Omega$ is understood as the set of all possible worlds, then $\Delta$ corresponds to the subset of all possible worlds in which $S_\Delta$ is true.  
The question for the quantum case is now straightforward: if the elements of $L(\h)$ correspond to propositions, what do these propositions express?

The most prominent early contribution to this question is undoubtedly the work of Birkhoff and von Neumann \cite{BirkhoffNeumann36}.
The backbone of this paper already appeared in \cite{Neumann55} where the idea of taking projections as propositions was introduced.
It is roughly motivated as follows.
If $A$ is an observable, and $\Delta$ a (measurable) subset of the reals, then the proposition $A\in\Delta$ is taken to be true if and only if the probability that a measurement of $A$ reveals some value in $\Delta$ equals one.
This is then the case if and only if the state of the system lies in $P_A^{\Delta}\mathcal{H}$, with $\mathcal{H}$ the Hilbert space describing the system, and $P_A$ the projection valued measure generated by $A$.
Or, in terms of density operators, if and only if $\mathrm{Tr}(\rho P_A^\Delta)=1$.

Of course this raises the question of what is meant by ``a measurement of $A$ reveals some value in $\Delta$''.
The most tempting answer is the one suggested by the notation $A\in\Delta$, i.e., that a measurement of $A$ reveals some value in $\Delta$ if and only if the observable $A$ has some value in $\Delta$, independent of the measurement.
A straightforward adoption of this realist reading is bound to run into difficulties posed by results such as the Kochen-Specker theorem \cite{KS67} or the Bell inequalities \cite{Bell64,Clauser69}.
A possible escape is to deny that $A\in\Delta$ expresses a proposition (or has a truth value) for every pair $(A,\Delta)$ at every instant.
This can be done by, for example, adopting the strong property postulate ($A$ has a value if and only if the system is in an eigen state of $A$).
Another possibility is to deny the bijection between observables and self-adjoint operators such as in Bohmian mechanics, where position observables play a privileged role.\footnote{In this approach, observables are associated with functions on the configuration space of particle positions. Every proposition $A\in\Delta$ then reduces to a proposition on particle positions. This reduction is contextual in a sense; the function associated with $A$ is not the same for every experimental setup used to measure $A$.}
Or one may suggest that these propositions do not obey the laws of classical logic.
This last option was the one suggested by Putnam in \cite{Putnam69}.
He argued that, apart from coinciding with a set of propositions about values possessed by observables, the lattice $L(\h)$ also describes the correct logic for these propositions.
However, this realist interpretation of the quantum logic $L(\h)$ is known to lead to problems \cite{Dummett76,Stairs83} and the consensus seems to be that there is no hope for this direction \cite{Maudlin05}. 

In contrast to these difficulties encountered in realist approaches to understanding $L(\h)$, it is often thought that an operationalist interpretation is unproblematic, and perhaps even straightforward:
\begin{quote}
If we put aside scruples about `measurement' as a primitive term in physical theory, and accept a principled distinction between `testable' and non-testable properties, then the fact that $L(\h)$ is not Boolean is unremarkable, and carries no implication about logic per se. 
Quantum mechanics is, on this view, a theory about the possible statistical distributions of outcomes of certain measurements, and its non-classical `logic' simply reflects the fact that not all observable phenomena can be observed simultaneously. -- Wilce \cite[\S 2]{Wilce12}
\end{quote}
The idea is that the task of interpreting the structure of $L(\h)$ is tied up with explicating the notion of measurement.
Then, since the operationalist takes the notion of measurement as primitive, the structure of $L(\h)$ may also be taken as primitive.
This line of reasoning seems unsatisfactory to me and I think that the questions of what the propositions in $L(\h)$ express, and what the logic is that governs them, are also meaningful and require study from an operationalist point of view.
And as a guide to such a logic one may take into account Bohr's demand that 
\begin{quote}
all well-defined experimental evidence, even if it cannot be analyzed in terms of classical physics, must be expressed in ordinary language making use of common logic. -- Bohr \cite{Bohr48}
\end{quote}
In this paper I continue the work done in \cite{Hermens12} to meet this demand.
Specifically, in section \ref{SQMsection} I first introduce a lattice of experimental propositions $S_{QM}$ which is then extended to an intuitionistic logic $L_{QM}$ in section \ref{LQMsection}.
These results comprise a summary of some of the results in \cite{Hermens12}.
Then, also in section \ref{LQMsection}, this logic is extended to a classical logic $CL_{QM}$ for experimental propositions.
In section \ref{QProbsec} R\'enyi's theory of conditional probability spaces \cite{Renyi55} is applied to $CL_{QM}$ and it is shown that the quantum probabilities as given by the Born rule can be reproduced.
This is followed by a short discussion on the issue that non-quantum probabilities are also allowed in this framework.
Finally, in section \ref{conclusionsection}, the results of this paper are evaluated and some ideas on what further role they may play in foundational debates are given.


\section{The lattice of experimental propositions}\label{SQMsection}
The idea to focus on the experimental side of QM is reminiscent of the approach of Birkhoff and von Neumann in \cite{BirkhoffNeumann36}.
They start by introducing `experimental propositions' which are identified with subsets of `observation spaces'.
An observation space is the Cartesian product of the spectra of a set of pairwise commuting observables.
However, once such a subset is associated with a projection operator in the familiar way, the experimental context (i.e., the set of pairwise commuting observables) is no longer being considered.
For example, one may have that $P_{A_1}^{\Delta_1}=P_{A_2}^{\Delta_2}$ even though $A_1$ and $A_2$ may be incompatible.
This indicates that whatever is encoded by propositions in $L(\h)$, they do not refer to actual experiments.

To avoid a cumbersome direct discussion of $L(\h)$, it is then appropriate to instead develop a logic of propositions that \emph{do} encode these experimental contexts.
As an elementary example of what Bohr may have had in mind when speaking of ``well-defined experimental evidence'' I take propositions of the form
\begin{equation}\label{reading}
  M_A(\Delta) =\text{``$A$ is measured and the result lies in $\Delta$.''}
\end{equation}
Of course, no logical structure can be derived for these propositions without some structural assumptions.
I will adopt the following two assumptions which I believe to be quite innocent\footnote{These assumptions are the ones also adopted in \cite{Hermens12} although there the second was implicit and the first was formulated to also apply to other theories.}:
\begin{itemize}
\item[\textbf{LMR}](Law-Measurement Relation) If $A_1$ and $A_2$ are two observables that can be measured together, and if $f$ is a function such that whenever $A_1$ and $A_2$ are measured together the outcomes $a_1$ and $a_2$ satisfy $a_1=f(a_2)$ (i.e., $f$ represents a law), then a measurement of $A_2$ alone also counts as a measurement of $A_1$ with outcome $f(a_2)$.\footnote{This assumption is reminiscent of the FUNC rule adopted in the Kochen-Specker theorem (c.f. \cite[p.121]{Redhead87}). However, FUNC is a much stronger assumption stating that whenever measurement outcomes for a pair of observables satisfy a functional relation, the values possessed by these observables must satisfy the same relationship. That is, FUNC is the conjunction of LMR together with the idea that measurements reveal pre-existing values.}
\item[\textbf{IEA}](Idealized Experimenter Assumption) Every experiment has an outcome, i.e., for every observable $A$ $M_A(\varnothing)$ is understood as a contradiction.
\end{itemize}
The following lemma describes the structure arising from these assumptions for QM. 
A more elaborate exposition may be found in \cite{Hermens12}.

\begin{lemma}
Under the assumptions LMR and IEA the set of (equivalence classes of) experimental propositions for quantum mechanics lead to the lattice 
\begin{equation}
  S_{QM}:=\{(\A,P)\:;\:\A\in\mathfrak{A},P\in L(\A)\backslash\{0\}\}\cup\{\bot\},
\end{equation}
where $\mathfrak{A}$ denotes the set of all Abelian operator algebras on $\h$ and $L(\A)$ the Boolean lattice $\A\cap L(\h)$.
The partial order is given by
\begin{equation}
  (\mathcal{A}_1,P_1)\leq(\mathcal{A}_2,P_2)~\text{iff}~P_1=0~\text{or}~\mathcal{A}_1\supset\mathcal{A}_2,P_1\leq P_2.
\end{equation}
\end{lemma}
\begin{proof}
A proposition $M_{A_1}(\Delta_1)$ implies $M_{A_2}(\Delta_2)$ if and only if a measurement of $A_1$ implies a measurement of $A_2$ and every outcome in $\Delta_1$ implies an outcome in $\Delta_2$ for the measurement of $A_2$ (or if $M_{A_1}{\Delta_1}$ expresses a contradiction).
By LMR, this is the case if and only if the Abelian algebra $\A_2$ generated by $A_2$ is a subalgebra of the Abelian algebra $\A_1$ generated by $A_1$ and $P_{A_1}^{\Delta_1}\leq P_{A_2}^{\Delta_2}$ (or $P_{A_1}^{\Delta_1}=0$).
This also establishes that $M_{A_1}(\Delta_1)$ and $M_{A_2}(\Delta_2)$ are equivalent if and only if $(\mathcal{A}_1,P_{A_1}^{\Delta_1})=(\mathcal{A}_2,P_{A_2}^{\Delta_2})$ (or both express a contradiction). 
\end{proof}

The meet on $S_{QM}$ is given by
\begin{equation}
  (\mathcal{A}_1,P_1)\wedge(\mathcal{A}_2,P_2)=\begin{cases}(\mathpzc{Alg}(\mathcal{A}_1,\mathcal{A}_2),P_1\wedge P_2),&\text{if}~[\mathcal{A}_1,\mathcal{A}_2]=0,P_1\wedge P_2\neq0,\\ \bot,&\text{otherwise,}\end{cases}
\end{equation}
where $P_1\wedge P_2$ is the meet on $L(\h)$.
It is consistent with the interpretation given to $M_A(\Delta)$ in the sense that this meet `distributes' over the `and' in ``$A$ is measured and the result lies in $\Delta$''.
In particular, it assigns a contradiction to a simultaneous measurement of two observables whenever their corresponding operators do not commute.

It is harder to interpret the join on this lattice as a disjunction.
When restricting to joins of two elements of $S_{QM}$ with the same algebra, the results seem intuitively correct.
In this case one has 
\begin{equation}\label{disjcom}
  (\mathcal{A},P_1)\vee(\mathcal{A},P_2)=(\mathcal{A},P_1\vee P_2).
\end{equation}
And indeed, ``$A$ is measured and the measurement reveals some value in $\Delta_1\cup\Delta_2$'' sounds like a good paraphrase of ``$A$ is measured and the measurement reveals some value in $\Delta_1$ or $A$ is measured and the measurement reveals some value in $\Delta_2$''.
But in more general cases problems arise.
For example, one has $(\mathcal{A}_1,1)\vee(\mathcal{A}_2,1)=(\mathcal{A}_1\cap\mathcal{A}_2,1)$.
Then, if the two algebras are completely incompatible (i.e. $\mathcal{A}_1\cap\mathcal{A}_2=\mathbb{C}1$), this equation states that ``$A_1$ is measured or $A_2$ is measured'' is a tautology even though one can consider situations in which neither is measured.

The example can further be used to show that $S_{QM}$ is not distributive.
Let $A_3$ be a third observable that is totally incompatible with both $A_1$ and $A_2$.
One then has
\begin{equation}
  (\mathcal{A}_3,1)\wedge((\mathcal{A}_1,1)\vee(\mathcal{A}_2,1))
  =(\mathcal{A}_3,1)\neq\bot
  =((\mathcal{A}_3,1)\wedge(\mathcal{A}_1,1))\vee((\mathcal{A}_3,1)\wedge(\mathcal{A}_2,1)).
\end{equation}
Hence, when it comes to providing a logic that fares well with natural language, $S_{QM}$ doesn't perform much better than the original $L(\h)$.


\section{Deriving the logics \texorpdfstring{$L_{QM}$}{LQM} and \texorpdfstring{$CL_{QM}$}{CLQM}}\label{LQMsection}
To solve the problem of non-distributivity of $S_{QM}$ the lattice has to be extended by formally introducing disjunctions which in general are stronger than the join on $S_{QM}$.
The outcome of this process is described by the following theorem.\footnote{Again, more details may be found in \cite{Hermens12}.}

\begin{theorem}
Formally introducing disjunctions to $S_{QM}$ while respecting \eqref{disjcom} results in the complete distributive lattice 
\begin{equation}
  L_{QM}=\left\{S:\mathfrak{A}\to L(\mathcal{H})\:;\:\substack{S(\mathcal{A})\in L(\A)\\ S(\A_1)\leq S(\A_2)~\text{whenever}~\A_1\subset\A_2}\right\}
\end{equation} 
with  partial order on $L_{QM}$ defined by
\begin{equation}
  S_1\leq S_2~\text{iff}~S_1(\mathcal{A})\leq S_2(\mathcal{A})\forall\mathcal{A}\in\mathfrak{A}
\end{equation} 
giving the meet and join
\begin{equation}\label{operations}
\begin{split}
  (S_1\vee S_2)(\mathcal{A})&=S_1(\mathcal{A})\vee S_2(\mathcal{A}),\\
  (S_1\wedge S_2)(\mathcal{A})&=S_1(\mathcal{A})\wedge S_2(\mathcal{A}).
\end{split}
\end{equation}
The embedding of $S_{QM}$ into $L_{QM}$ is given by\footnote{For notational convenience $(\A,0)$ is identified with $\bot\in S_{QM}$.}
\begin{equation}\label{imbed}
  i:(\mathcal{A},P)\mapsto S_{(\mathcal{A},P)},~S_{(\mathcal{A},P)}(\mathcal{A}'):=\begin{cases}P,&\mathcal{A}\subset\mathcal{A}',\\0,&\text{otherwise}.\end{cases}
\end{equation}
\end{theorem}
\begin{proof}
The key to the extension of $S_{QM}$ to $L_{QM}$ lies in noting that the embedding \eqref{imbed} satisfies the relation
\begin{equation}\label{relation}
  S=\bigvee_{\A\in\mathfrak{A}}i(\A,S(\A)).
\end{equation}
This gives the intended reading of elements of $L_{QM}$, namely, as a disjunction of experimental propositions.
By LMR it follows that $i$ preserves the interpretation of elements of $S_{QM}$.
The conjunctions on $L_{QM}$ further agree with those on $S_{QM}$:
\begin{equation}
  S_{(\A_1,P_1)}\wedge S_{(\A_2,P_2)}= S_{(\A_1,P_1)\wedge(\A_2,P_2)}.
\end{equation}
Further, $L_{QM}$ respects \eqref{disjcom} while introducing `new' elements for other disjunctions, i.e.,
\begin{equation}
  S_{(\A_1,P_1)}\vee S_{(\A_2,P_2)}\leq S_{(\A_1,P_1)\vee(\A_2,P_2)},
\end{equation}
with equality if and only if either $P_1$ or $P_2$ equals $0$, or $\A_1=\A_2$. 
Finally, that $L_{QM}$ is a complete distributive lattice follows from \eqref{operations} and observing that $L(\mathcal{A})$ is a complete distributive lattice for every $\mathcal{A}\in\mathfrak{A}$.
\end{proof}

The lattice $L_{QM}$ is turned into a Heyting algebra by obtaining the relative pseudo-complement in the usual way:
\begin{equation}
  S_1\to S_2=\bigvee\left\{S\in L_{QM}\:;\:S\wedge S_1\leq S_2\right\}.
\end{equation}
And negation is introduced by $\neg S:=S\to\bot$ thus obtaining an intuitionistic logic for reasoning with experimental propositions.
The logic $L_{QM}$ is also a proper Heyting algebra and, in fact, is radically non-classical.
That is, the law of excluded middle only holds for the top and bottom element.
This is easily checked by observing that $S(\mathbb{C}1)=1$ if and only if $S=\top$. 
It then follows that $S\vee\neg S=\top$ if and only if $S=\top$ or $\neg S=\top$.

The strong non-classical behavior can be understood by noting that the logic is entirely built up from propositions about performed measurements.
The negation of $M_A(\Delta)$ is identified with other propositions about measurements that exclude the possibility of $M_A(\Delta)$ but leave open the option of no experiment having been performed:
\begin{equation}
  \neg S_{(\A,P)}=\bigvee \left\{S_{(\A',P')}\:;\: [\A',\A]\neq 0~\text{or}~P'\wedge P=0\right\}.
\end{equation}
An unexcluded middle thus presents itself as the proposition $\neg M_A$=``$A$ is not measured''.
It was suggested in \cite[\S2]{Hermens12} that incorporating such propositions could lead to a classical logic.
I will now show that this is indeed the case.

It is tempting to start the program for a classical logic for experimental propositions by formally adding the suggested propositions $\neg M_A$ and building on from there.
However, it turns out that it is much more convenient to introduce new elementary experimental propositions, namely
\begin{equation}\label{prov}
  M_A^!(\Delta)=\text{``$A$ is measured with result in $\Delta$, and no finer grained measurement has been performed''}.
\end{equation}
The notion of fine graining here relies on LMR.
Another way of expressing the conjunct added is to state that ``only $A$ and all the observables implied to be measured by the measurement of $A$ have been measured''.
With this concept the following can now be shown.\footnote{In the remainder of this paper it is tacitly assumed that $\h$ is finite-dimensional.}

\begin{theorem}
The classical logic for experimental propositions is given by\footnote{The exclamation mark is merely a notational artifact introduced to discern the elements of $S^!_{QM}$ from those of $S_{QM}$. That is, without the exclamation mark, $S^!_{QM}$ would be a subset of $S_{QM}$.}
\begin{equation}
  CL_{QM}=\mathcal{P}(S^!_{QM}),~S^!_{QM}:=\left\{(\A^!,P)\:;\:\A\in\mathfrak{A}, P\in L_0(\A)\right\},
\end{equation}
where $L_0(\A)$ denotes the set of atoms of the lattice $L(\A)$. 
\end{theorem}
\begin{proof}
Being defined as a power set, it is clear that $CL_{QM}$ is a classical propositional lattice. 
So it only has to be shown that every element signifies an experimental proposition, and that it is rich enough to incorporate the propositions already introduced by $L_{QM}$. 
The first issue is readily established by the identification
\begin{equation}\label{ident}
  M_A^!(\Delta)\mapsto \{(\A^!,P)\in S^!_{QM}\:;\:P\leq \mu_A(\Delta)\},
\end{equation}
with $\A$ the algebra generated by $A$.
Identifying $\A$ with $A$ is again justified by LMR, and IEA is incorporated by the fact that, whenever $\Delta$ contains no elements in the spectrum of $A$, $M_A^!(\Delta)$ is identified with the empty set.
Thus \eqref{ident} ensures that every element of $CL_{QM}$ is understood as a disjunction of experimental propositions, i.e., every singleton set in $CL_{QM}$ encodes a proposition of the form $M_A^!(\{a\})$.

That $CL_{QM}$ is indeed rich enough follows by observing that the map
\begin{equation}
  c:L_{QM}\to CL_{QM},~S\mapsto\bigcup_{\A\in\mathfrak{A}}\left\{(\A^!,P)\in S^!_{QM}\:;\:P\leq S(\A)\right\}
\end{equation}
is an injective complete homomorphism.
\end{proof}

Some additional reflections are helpful to get a grip on $CL_{QM}$.
First one may note that the map $(c\circ i):S_{QM}\to CL_{QM}$, encoding propositions of the form $M_A(\Delta)$, acts as
\begin{equation}\label{ident2}
  c\circ i:(\A,P)\mapsto \left\{(\A'^!,P')\in S^!_{QM}\:;\:(\A',P')\leq (\A,P)\right\}.
\end{equation}
Thus the provision ``and no finer grained measurement has been performed'' that distinguishes $M_A^!(\Delta)$ from $M_A(\Delta)$ precisely excludes all the $(\A'^!,P')\in (c\circ i)(\A,P) $ with $\A'\neq\A$, i.e., the finer grained measurements (c.f. \eqref{ident} and \eqref{ident2}). 

The non-classicality of $L_{QM}$ is also explicated by $CL_{QM}$.
For any $(\A,P)\in S_{QM}$, the complement of 
\begin{equation}
  c(S_{(\A,P)}\vee\neg S_{(\A,P)})
\end{equation}
consists of all the $(\A'^!,P')\in S^!_{QM}$ with $\A'\subsetneq\A$ and $P'\leq P$.
This is because only the added provision in \eqref{prov} makes $(\A'^!,P')$ exclude the measurement of $A$, while within $L_{QM}$ there is nothing to denote a possible incompatibility for $A$ and $A'$.
A proposition in $CL_{QM}$ of special interest in this context is the singleton $\{(\mathbb{C}1,1)\}$ which expresses that only the trivial measurement is performed.
In light of IEA this is just the same as saying that no measurement is performed; the most typical proposition in $CL_{QM}$ not present in $L_{QM}$.  

To conclude this section I compare the constructions of $L_{QM}$ and $CL_{QM}$ to another way of extending $S_{QM}$ to a distributive lattice, namely, by using the Bruns-Lakser theory of distributive hulls \cite{BrunsLakser70} as advocated in \cite{Coecke02}.
In this approach $S_{QM}$ is embedded into the lattice of its distributive ideals\footnote{A distributive ideal is a subset $I\subset S_{QM}$ that is downward closed and further contains the join $\bigvee\{(\A,P)\in J\}$ for every subset $J\subset I$ that satisfies $\left(\bigvee_{(\A,P)\in J}(\A,P)\right)\wedge(\A',P')= \bigvee_{(\A,P)\in J}\left((\A,P)\wedge(\A',P')\right)$ for all $(\A',P')\in S_{QM}$.} $\mathcal{DI}(S_{QM})$ by taking each element to its downward set, i.e.,
\begin{equation}
  (\A,P)\mapsto\downarrow(\A,P):=\{(\A',P')\in S_{QM}\:;\:(\A',P')\leq(\A,P)\}.
\end{equation}

Generally, the abstract nature of the extension $\mathcal{DI}(L)$ for some non-distributive lattice $L$ prevents a straightforward identification of distributive ideals with propositions.
In the present situation, for example, it is not immediately clear whether the resulting lattice operations behave consistently with respect to the intended reading of the elements $\downarrow(\A,P)$ given by \eqref{reading}. 
Fortunately, this issue can be resolved for $S_{QM}$ by noting that the resulting lattice of its distributive ideals is isomorphic to the power set of the atoms in $S_{QM}$:
\begin{equation}
  \mathcal{DI}(S_{QM})\simeq\mathcal{P}(S^*_{QM}),~S^*_{QM}:=\{(\A,P)\in S_{QM}\:;\:\A\in \mathfrak{A}_{\mathrm{max}},P\in L_0(\A)\},
\end{equation}
where $\mathfrak{A}_{\mathrm{max}}$ is the set of maximal Abelian algebras.
Here then one sees that the proposition identified with $(\A,P)\in S_{QM}$ is understood in  $\mathcal{P}(S^*_{QM})$ as a disjunction over all possible \emph{complete} measurements compatible with a measurement of $A$. 
That is, $M_A(\Delta)$ in this setting may be rephrased as ``some complete measurement $A'$ has been performed with outcome $a'$ which implies $M_A(\Delta)$''.
Although this provides a consistent reading, the underlying presupposition that every measurement in fact constitutes a complete measurement seems somewhat reckless.


\section{Quantum logic and probability}\label{QProbsec}
If the logic $CL_{QM}$ is to play any significant role in quantum mechanics, it has to be connected with quantum probability theory.
That is, one has to know how the Born rule can be expressed using the language of $CL_{QM}$.
Ideally this consists of two parts.
First, to show that probability functions on $CL_{QM}$ can be introduced that are in accordance with the Born rule.
And secondly, to show that the Born rule comes out as a necessary consequence holding for all probability functions on $CL_{QM}$ (i.e., a result similar to Gleason's theorem for $L(\h)$ (\cite{Gleason57})).  
In this section some first investigations in this direction are made.

The first approach for probabilities on $CL_{QM}$ undoubtedly is to note that $(S^!_{QM},CL_{QM})$ is a measure space.
Hence, it is tempting to state that all probability functions are simply all probability measures on this space.
However, this does not do justice to the quantum formalism.
For example, one possible probability measure is the Dirac measure peaked at $\{(\mathbb{C}1,1)\}$ expressing certainty that no measurement will be performed.
This is likely to be too minimalistic even for the hardened instrumentalist. 
That is, although one may adhere to the idea that `unperformed experiments have no results' \cite{Peres78}, a central idea in QM is that possible outcomes for measurements do have probabilities irrespective of whether the measurements are performed.
Indeed, it is a trait of QM that certain conditional probabilities have definite values even if the propositions on which one conditions have prior probability zero.
The Dirac measure does not provide these conditional probabilities.

To get around this it is natural to take conditional probabilities as primitive in the quantum case.\footnote{Quantum mechanics has even been used as an example to advocate that conditional probability should be considered more fundamental than unconditional probability, c.f. \cite{Hajek03}.}
For this approach the axiomatic scheme of R\'enyi is an appropriate choice.\footnote{C.f., \cite{Renyi55}. A similar axiomatic approach was developed independently by Popper \cite[A *ii--*v]{Popper59}. The present presentation is taken from \cite{Spohn86}.}

\begin{definition}
Let $(\Omega,\mathcal{F})$ be a measurable space and $\mathcal{C}\subset\mathcal{F}\backslash\{\varnothing\}$ be a non-empty set, then $(\Omega,\mathcal{F},\mathcal{C},\pee)$ is called a \emph{conditional probability space (CPS)} if $\pee:\mathcal{F}\times\mathcal{C}\to[0,1]$ satisfies
\begin{enumerate}
\item for every $C\in\mathcal{C}$ the function $A\mapsto\pee(A|C)$ is a probability measure on $(\Omega,\mathcal{F})$,
\item for all $A,B,C\in\mathcal{F}$ with $C,B\cap C\in\mathcal{C}$ one has
\begin{equation}
  \pee(A\cap B|C)=\pee(A|B\cap C)\pee(B|C).
\end{equation}
\end{enumerate}
If in addition to these criteria $\mathcal{C}$ is closed under finite unions, the space is called an \emph{additive CPS}, and if $\mathcal{C}=\mathcal{F}\backslash\{\varnothing\}$ it is called a \emph{full CPS}.
\end{definition}

In the present situation it is straightforward to take $(S^!_{QM},CL_{QM})$ as the measure space.
For the set of conditions a modest choice is to take the set of propositions expressing the performance of a measurement: 
\begin{equation}
  \mathcal{C}_{QM}:=\left\{F_{\A}\:;\:\A\in\mathfrak{A}\right\},~F_{\A}:=(c\circ i)(\A,1).
\end{equation}
For this set of conditions the Born rule can easily be captured with the CPS $(S^!_{QM},CL_{QM},\mathcal{C}_{QM},\pee_{\rho})$ with the probability function given by
\begin{equation}\label{qmcpss}
  \pee_\rho(F|F_{\A})
  :=\sum_{\left\{\substack{P\in L(\h)\:;\\ (\A^!,P)\in F}\right\}}\mathrm{Tr}(\rho P),\\
\end{equation}
for $F\in CL_{QM}$ and with $\rho$ a density operator on $\h$.\footnote{Another possibility is to take for the conditions the set of propositions expressing the performance of a measurement excluding the possibility of a finer grained measurement: 
\begin{equation*}
\mathcal{C}_{QM}^!:=\left\{F^!_{\A}\:;\:\A\in\mathfrak{A}\right\},~F^!_{\A}:=\{(\A^!,P)\in S^!_{QM}\:;\:P\in L(\A)\}.
\end{equation*}
Also in this setting the Born rule can be recovered by taking the CPS $(S^!_{QM},CL_{QM},\mathcal{C}^!_{QM},\pee^!_{\rho})$ with the probability function given by
\begin{equation*}
  \pee_\rho^!(F|F^!_{\A})
  :=\sum_{\left\{\substack{P\in L(\h)\:;\\ (\A^!,P)\in F}\right\}}\mathrm{Tr}(\rho P)=\pee_\rho(F|F_{\A}).
\end{equation*}
Much of the remainder of this section can also be phrased in terms of this CPS.
I will however restrict myself to \eqref{qmcpss} for the sake of clarity.} 

It would however be more interesting to take $\mathcal{C}=CL_{QM}\backslash\{\varnothing\}$ for the set of conditions.
At the present it is not clear to me if this is possible, as extending the set of conditions seems a non-trivial matter.
There is a positive result on this matter however, namely, the CPS
$(S^!_{QM},CL_{QM},\mathcal{C}_{QM},\pee_{\rho})$ can be extended to
an additive CPS.
The proof for this claim relies on the following definition and theorem.

\begin{definition}
A family of measures $(\mu_i)_{i\in I}$ on $(\Omega,\mathcal{F})$ is called a \emph{generating family of measures for the CPS}  $(\Omega,\mathcal{F},\mathcal{C},\pee)$ if and only if for all $C\in\mathcal{C}$ there exists $i\in I$ such that $0<\mu_i(C)<\infty$ and for all $F\in\mathcal{F}$, for all $C\in\mathcal{C}$ and for all $i\in I$ with $0<\mu_i(C)<\infty$
\begin{equation}
  \pee(F|C)=\frac{\mu_i(F\cap C)}{\mu_i(C)}.
\end{equation}
The family is called \emph{dimensionally ordered} if and only if there is a total order $<$ on $I$ such that for all $F\in\mathcal{F}$ and $i,j\in I$: $i<j$ $\mu_j(F)=0$ whenever $\mu_i(F)<\infty$. 
\end{definition}

\begin{theorem}[Cs\'asz\'ar 1955 \cite{Csaszar55}]
A CPS $(\Omega,\mathcal{F},\mathcal{C},\pee)$ can be extended to an additive CPS $(\Omega,\mathcal{F},\mathcal{C}',\pee')$ (i.e., $\mathcal{C}\subset\mathcal{C}'$, and $\pee'$ acts as $\pee$ when restricted to $\mathcal{F}\times\mathcal{C}$) if and only if it is generated by a dimensionally ordered family of measures on $(\Omega,\mathcal{F})$.
\end{theorem}

Thus to show that the CPS given by $\pee_\rho$ in \eqref{qmcpss} can be extended to an additive CPS it suffices to give a dimensionally ordered family of measures  that generates it.
The construction of this family relies on the following dimensionality function
\begin{equation}
\begin{split}
  d:\mathfrak{A}\to\mathbb{N},&~d(\A):=\# L_0(\A),\\
  d:CL_{QM}\to\mathbb{N},&~d(F):=\min\{d(\A)\:;\:(\A^!,P)\in F\}.
\end{split}
\end{equation}
That is, $d$ picks out the coarsest measurements in $F$. 
In particular, $d(F_{\A})=d(\A)$.
The family of measures $\left(\mu_i^\rho\right)_{i\in\mathbb{N}}$ is now defined as
\begin{equation}
   \mu_i^\rho(F):=\delta_{id(F)}\sum_{\left\{\substack{P\in L(\h)\:;\:\exists\A\in\mathfrak{A}\text{ s.t.}\\(\A^!,P)\in F, d(\A)=i}\right\}}\mathrm{Tr}(\rho P),
\end{equation}
where $\delta_{id(F)}$ denotes the Kronecker delta.
This Kronecker delta warrants that the family is indeed dimensionally ordered where the order is the usual one on $\mathbb{N}$.
Using that
\begin{equation}
   \mu_i^\rho(F_{\A})=\delta_{id(\A)}
\end{equation}
one may easily check that this is indeed a generating family for the CPS $(S^!_{QM},CL_{QM},\mathcal{C}_{QM},\pee_\rho)$.

Another open issue with the CPS approach to quantum probabilities is that of deriving the Born rule, i.e., to derive the conditions under which $\pee$ is of the form $\pee_{\rho}$.
A particular difficulty here is that quantum probabilities are non-contextual, a property that is not expected to hold in general for an admissible probability function on $(S^!_{QM},CL_{QM})$.
That is, one may have that for $\A_1,\A_2\supset\A$ the equality
\begin{equation}\label{ineqB}
  \pee\left(\{(\A^!,P)\}\middle|F_{\A_1}\right)=\pee\left(\{(\A^!,P)\}\middle|F_{\A_2}\right),
\end{equation}
does not hold, although it is satisfied for the Born rule \eqref{qmcpss}.

Conversely, if in addition to \eqref{ineqB} the equality
\begin{equation}
  \pee\left(\{(\A_1^!,P)\}\middle|F_{\A}\right)=\pee\left(\{(\A_2^!,P)\}\middle|F_{\A}\right)
\end{equation}
were also to hold whenever $\A_1,\A_2\subset\A$, the Born rule would follow.\footnote{Provided that the dimension of $\h$ is greater than $2$.}
This is because then the function
\begin{equation}
  \pee_{QM}:L(\h)\to[0,1],~\pee_{QM}(P):=\pee\left(\{(\A^!,P)\}\middle|F_{\A'}\right)
\end{equation}
with $\A'\supset\A$ is well defined (i.e., independent of the choice of $\A$ and $\A'$) and satisfies
\begin{itemize}
\item $\pee_{QM}(1)=1$,
\item $\pee_{QM}(P_1\vee P_2\vee\ldots)=\pee(P_1)+\pee(P_2)+\ldots$ \quad whenever $P_iP_j=0$ for all $i\neq j$.
\end{itemize} 
Hence, Gleason's theorem applies.

Now, before coming to an end it is interesting to take a quick look at the connection between possible extensions of the CPSs and Bell inequalities.
Consider the standard EPR-Bohm situation where Alice can choose between two possible measurements $A^A_1$  and $A^A_2$, whilst Bob can choose between two possible measurements $A^B_1$ and $A^B_2$. 
Let $\A_{ij}$ denote the algebra generated by $\mathcal{A}^A_i$ and $\mathcal{A}^B_j$. 
It is further presupposed that each measurement has two possible outcomes associated with the projections $P^X_{i0}$ and $P^X_{i1}$ in $\A^X_ i$.

The CPSs given by \eqref{qmcpss} aren't rich enough to evaluate possible derivations for a Bell inequality.
For that, the set of conditions has to be extended to
\begin{equation}
  \mathcal{C}_{QM}':=\left\{F_{(\A,P)}\:;\:(\A,P)\in S_{QM}\right\},~F_{(\A,P)}:=(c\circ i)(\A,P).\\
\end{equation}
With these conditions the requirements of outcome independence (OI) and parameter independence (PI) from which the Bell inequality can be derived \cite{Jarrett84,Shimony84} can be stated:
\begin{equation}
\begin{split}
  \pee\left(\{(\A^{A!}_i,P^A_{i0})\}\middle|F_{(\A_{i1},1)}\right)
  &=
  \pee\left(\{(\A^{A!}_i,P^A_{i0})\}\middle|F_{(\A_{i2},1)}\right)\\
  \pee\left(\{(\A^{B!}_i,P^B_{i0})\}\middle|F_{(\A_{1i},1)}\right)
  &=
  \pee\left(\{(\A^{B!}_i,P^B_{i0})\}\middle|F_{(\A_{2i},1)}\right) 
\end{split} \tag{PI}
\end{equation}
\begin{equation}
\begin{split}
  \pee\left(\{(\A^{A!}_i,P^A_{i0})\}\middle|F_{(\A_{ij},1)}\right)
  &=
  \pee\left(\{(\A^{A!}_i,P^A_{i0})\}\middle|F_{(\A_{ij},P^B_{jk})}\right)\\
  \pee\left(\{(\A^{B!}_i,P^B_{i0})\}\middle|F_{(\A_{ji},1)}\right)
  &=
  \pee\left(\{(\A^{B!}_i,P^B_{i0})\}\middle|F_{(\A_{ji},P^A_{jk})}\right).
\end{split} \tag{OI}
\end{equation}
The violation of Bell inequalities in QM shows that OI must fail for any extension of $\pee_{\rho}$ to $\mathcal{C}_{QM}'$.
This may be indicative of the idea that the development of \emph{any} full CPS $(S^!_{QM},CL_{QM},CL_{QM}\backslash\{\varnothing\},\pee)$ is a non-trivial matter.\footnote{Note that I'm thinking here of interesting cases. A trivial CPS is easily constructed. Just fix one $(\A !,P)\in S^!_{QM}$ and take $\pee(F|C)=\begin{cases}1&(\A^!,P)\in F,\\0& \text{otherwise},\end{cases}$ \quad for all $C\in CL_{QM}\backslash\{\varnothing\}$.}


\section{Conclusion and outlook}\label{conclusionsection}
With the construction of $CL_{QM}$ a concrete formalism is provided that adheres to the Bohrian view that also for quantum mechanics experimental reports must be communicable in unambiguous common language.
However, the association with the Copenhagen doctrine ends there.
Although the notion of measurement is taken as primitive here, no demand is made for the necessity of this.
That is, the current program should not be seen as a method to dismiss the measurement problem, but merely as a way of making sense of the formalism without endorsing a particular solution to the problem.
In this paper one obstacle in such investigations has been overcome, namely, the non-distributivity of quantum logic.
In return a less abstract obstacle has emerged, namely, the justification of the Born rule in quantum mechanics.
To get a further grip on this issue it may be fruitful future work to seek connections with other programs in quantum foundations such as generalized probability theory \cite{Barrett07} or the Bayesian approach to quantum probabilities \cite{Fuchs10}.


\section*{Acknowledgments}
I would like to thank the people at the QPL 2012 conference in Brussels and at the first Quantum Toposophy workshop in Nijmegen for useful questions and discussions. 
And further J. W. Romeijn for insisting on the use of classical logic.
This work was supported by the NWO (Vidi project nr. 016.114.354).

\bibliographystyle{eptcs}    
\bibliography{qplrefs}   

\end{document}